\documentclass{article}

\usepackage{arxiv}

\usepackage[utf8]{inputenc} %
\usepackage[T1]{fontenc}    %
\usepackage{url}            %
\usepackage{nicefrac}       %
\usepackage{microtype}      %
\usepackage{lipsum}		%
\usepackage{doi}

\usepackage{cite}
\usepackage{amsmath,amsthm, amssymb,amsfonts}
\usepackage{graphicx}
\usepackage{comment}
\usepackage{subcaption}
\usepackage{mwe}
\usepackage{hyperref}
\usepackage{epstopdf}

\usepackage[utf8]{inputenc}
\usepackage{algorithm}
\usepackage{algpseudocode}

\usepackage{multirow}
\usepackage{booktabs}
\usepackage{tabularx}
\usepackage{textcomp}
\newcolumntype{x}[1]{>{\centering\arraybackslash\hspace{0pt}}p{#1}}

\newcolumntype{C}{>{\centering\arraybackslash}X}

\newfloat{lstfloat}{tb}{lop}
\floatname{lstfloat}{Listing}

\usepackage{listings} %
\usepackage{seqsplit} %
\usepackage{xcolor} %

\newtheorem{theorem}{Theorem}

\newtheorem{corollary}[theorem]{Corollary}

\usepackage{eucal}

\let\svthefootnote\thefootnote
\newcommand\freefootnote[1]{%
  \let\thefootnote\relax%
  \footnotetext{#1}%
  \let\thefootnote\svthefootnote%
}

\makeatletter
\def\lst@lettertrue{\let\lst@ifletter\iffalse}
\makeatother

\lstdefinelanguage{json}{
    basicstyle=\normalfont\ttfamily,
    numbers=left,
    numberstyle=\scriptsize,
    stepnumber=1,
    numbersep=8pt,
    showstringspaces=false,
    breaklines=true,
    frame=lines,
    backgroundcolor=\color{background},
    stringstyle=\color{red},
    keywordstyle=\color{blue},
    commentstyle=\color{green},
}

\definecolor{green}{rgb}{0,0.6,0}
\definecolor{background}{rgb}{0.9,0.9,0.9}
\definecolor{blue}{rgb}{0.58,0,0.82}
\definecolor{red}{rgb}{0.95,0.95,0.92}

\title{Implementation Study of Cost-Effective Verification for Pietrzak's VDF in Ethereum Smart Contracts}

\author{%
  Suhyeon Lee  \textsuperscript{\textsection}   \\
  School of Cybersecurity\\
  Korea University \\ 
  \texttt{orion-alpha@korea.ac.kr}\\
  Tokamak Network \\
  \texttt{suhyeon@tokamak.network} \\
  \And
  Euisin Gee \textsuperscript{\textsection} \\
  School of Cybersecurity\\
  Korea University \\ 
  \texttt{usgee@korea.ac.kr}\\
  Onther \\
  \texttt{
  justin.g@onther.io} \\
  \And
  Junghee Lee \\
  School of Cybersecurity\\
  Korea University \\
  \texttt{j\_lee@korea.ac.kr} \\
}

\renewcommand{\shorttitle}{Cost-Effective Pietrzak VDF Verification in Smart Contracts}

\hypersetup{
pdftitle={Implementation Study of Cost-Effective Verification for Pietrzak's Verifiable Delay Function in Ethereum Smart Contracts},
pdfsubject={CoRR},
pdfauthor={Suhyeon Lee, Euisin Jee, Junghee Lee},
pdfkeywords={Ethereum, blockchain security, cryptographic protocol, decentralized application, gas optimization, verifiable delay function},
}

\begin{document}
\maketitle

\begingroup\renewcommand\thefootnote{\textsection}
\footnotetext{These authors contributed equally to this work.}
\endgroup

\begin{abstract}
    
    Verifiable Delay Function (VDF) is a cryptographic concept that ensures a minimum delay before output through sequential processing, which is resistant to parallel computing. One of the significant VDF protocols academically reviewed is the VDF protocol proposed by Pietrzak. However, for the blockchain environment, the Pietrzak VDF has drawbacks including long proof size and recursive protocol computation. In this paper, we present an implementation study of Pietrzak VDF verification on Ethereum Virtual Machine (EVM). We found that the discussion in the Pietrzak's original paper can help a clear optimization in EVM where the costs of computation are predefined as the specific amounts of gas. In our results, the cost of VDF verification can be reduced from 4M to 2M gas, and the proof length can be generated under 8 KB with the 2048-bit RSA key length, which is much smaller than the previous expectation.

\end{abstract}


\section{Introduction}
\label{section: introduction}

The Verifiable Delay Function (VDF)'s major role is to guarantee a specific delay for evaluation. It can play an important role in security applications. The most notable application of blockchain using VDF is a distributed randomness beacon. VDFs do not directly generate a random number. Rather, they help to guarantee the liveness of a random number generation scheme, Commit-Reveal, by recovering a random number when the reveal phase fails. This has led Ethereum, one of the largest cryptocurrencies, to include VDFs in their future roadmap for the RANDAO block proposer selection mechanism.

In recent years, several VDF mechanisms have been proposed \cite{wesolowski2019efficient, pietrzak2019simple, lenstra2017trustworthy, khovratovich2022minroot}; however, the concept of VDFs remains relatively new. This indicates that we need comprehensive research on VDFs across various dimensions. For example, recently found parallel computation mechanisms for VDF evaluation \cite{biryukov2024cryptanalysis, leurent2023analysis} show that we need more research on the squaring computation. On the other hand, there is a scarcity of implementation research \cite{attias2020implementation} on VDFs. In this paper, we focus on VDF implementation on EVM, one of the leading blockchain environments.

Our goal is to efficiently implement VDF verification within smart contracts, providing a foundation for EVM security applications utilizing VDFs. Specifically, we focus on the Pietrzak VDF \cite{pietrzak2019simple} for two key reasons: First, the Wesolowski VDF \cite{wesolowski2019efficient}, another major VDF, requires a hash-to-prime function, which is challenging to implement directly in smart contracts (see \ref{appendix: hash-to-prime}). Second, we hypothesize that the performance optimization technique discussed by Pietrzak \cite{pietrzak2019simple} could significantly reduce the overall cost within the EVM environment. As a result of this focus, our research demonstrates through theoretical analysis and experiments that Pietrzak VDF verification can be effectively utilized within EVM-compatible smart contracts.

The contributions of this paper are as follows:

\begin{itemize} 
    \item To the best of our knowledge, this research presents the first academic study of Pietrzak VDF implementation on smart contracts. 
    \item Through theoretical analysis, we demonstrated the existence of a unique parameter that reduces the proof length and the execution gas fee, making the implementation more cost-efficient.
    \item We demonstrated that the Pietrzak VDF could be utilized in the Ethereum environment with under 8 KB calldata size for a 2048-bit RSA key length. This calldata size represents the serialized data required for VDF verification and highlights the feasibility of implementing the Pietrzak VDF on Ethereum smart contracts, contrary to previous studies' expectations. 
\end{itemize}

The rest of the paper is organized as follows: Section \ref{section: related works} reviews related works on VDFs. Section \ref{section: target} describes the target algorithm and experiment environment. In Section \ref{section: measure}, we measure the gas costs for VDF verification on the EVM and formulate them into a function using regression. Section \ref{section: result} analyzes the optimal parameters for VDF verification costs and presents experimental results that demonstrate the theoretical analysis. Section \ref{section: discussion} addresses challenges in implementing VDFs on Ethereum and explores potential directions for research and optimization. Finally, Section \ref{section: conclusion} concludes the paper by summarizing our findings and outlining future research plans. The source code for our implementation is available in \ref{appendix: code availability}.

\section{Related Works} \label{section: related works}

Decentralized, bias-resistant randomness is a cornerstone of blockchain security, driving extensive research into Distributed Randomness Beacon (DRB) protocols. Recent Systematization of Knowledge (SoK) studies by Raikwar and Gligoroski \cite{raikwar2022sok}, Choi et al. \cite{choi2023sok}, and Kavousi et al. \cite{kavousi2024sok} provide comprehensive taxonomies of DRBs, analyzing them in terms of cryptographic primitives, interactivity, unbiasability, unpredictability, and scalability. These surveys underscore the critical role of publicly verifiable randomness in preventing adversarial bias or manipulation in decentralized applications.

One cryptographic mechanism that has gained significant traction in DRB protocols is the Verifiable Delay Function (VDF). The concept of VDFs builds on earlier advances in time-delayed cryptography. Rivest et al. introduced the “Time-Lock Puzzle” \cite{rivest1996time}, which envisioned recursive computation for delayed output but lacked efficient public verifiability. Mahmoody et al. later proposed Proof of Sequential Work (PoSW) \cite{mahmoody2013publicly}, thereby formalizing sequential tasks in cryptographic protocols. Lenstra and Wesolowski’s Sloth \cite{lenstra2017trustworthy} is often regarded as a precursor to modern VDFs, which were later formalized by Boneh et al. \cite{boneh2018verifiable} into efficient and publicly verifiable delay primitives.

Following the initial formalization of VDFs, various constructions have emerged. Wesolowski’s \cite{wesolowski2019efficient} and Pietrzak’s \cite{pietrzak2019simple} designs are particularly notable for their strong security proofs and computational efficiency. Beyond DRBs, VDFs have also found applications in non-interactive cryptographic timestamping \cite{landerreche2020non} and DoS mitigation \cite{raikwar2021non}, underscoring their versatility. Pietrzak’s VDF, in particular, has been adapted in several ways: Hoffmann and Pietrzak \cite{hoffmann2024watermarkable} extend it to watermarkable and zero-knowledge settings, Motepalli and Jacobsen \cite{motepalli2024delay} apply it to blockchain bootstrapping, and it used the Pietrzak implementation by Chia Network \cite{chiaGreenPaper}. However, these works focus primarily on conceptual or protocol-level enhancements, rather than on optimizing on-chain performance in environments like the Ethereum Virtual Machine (EVM).

In addition to these foundational VDF-based designs, other DRB solutions have been proposed that rely on various cryptographic primitives to achieve robust randomness on-chain. For example, RandRunner \cite{schindler2021randrunner} employs trapdoor VDFs to ensure strong uniqueness, Bicorn \cite{bicorn} incorporates VDFs in an optimistically efficient DRB, and RandChain \cite{han2020randchain} leverages sequential Proof-of-Work to decentralize randomness generation. Raikwar \cite{raikwar2022competitive} further examines competitive DRB protocols, highlighting different strategies for achieving resilience against adversarial manipulation. Collectively, these systems reinforce the importance of secure, verifiable primitives that balance theoretical soundness with practical feasibility in blockchain contexts.

Later proposals, such as Veedo and Minroot \cite{khovratovich2022minroot}, target specific blockchain ecosystems (e.g., Starknet or RANDAO). At the same time, recent research by Biryukov et al. \cite{biryukov2024cryptanalysis} and Leurent \cite{leurent2023analysis} indicates that distributed modular squaring might be possible, challenging the assumption of a strictly sequential computation. This highlights the need to reconcile theoretical security assumptions with real-world performance constraints in VDF-based DRBs.

Despite increasing interest in VDFs, implementation-focused investigations remain sparse. One of the few in-depth evaluations is by Attias et al. \cite{attias2020implementation}, who examined both Pietrzak’s and Wesolowski’s VDFs—covering evaluation, proof generation, and verification—and noted that large proof sizes could be problematic for blockchain environments. Choi et al. \cite{bicorn} later integrated Wesolowski’s VDF into Ethereum commit-reveal schemes but did not provide a cost-optimization analysis or a large-scale feasibility study.

To address this gap, our work examines the on-chain performance of Pietrzak’s VDF on Ethereum smart contracts. We concentrate on gas costs, proof sizes, and optimization strategies that render VDF-based solutions viable. By offering a comprehensive feasibility analysis for Pietrzak’s VDF in real-world blockchain environments, we extend existing DRB research and demonstrate that VDFs can be efficiently deployed within the EVM.

\section{Implementation Target and Environment} \label{section: target}

In this section, we provide the target algorithm and environment. Firstly, we define the Pietrzak VDF algorithm we aim to implement. Secondly, we describe the experiment environment for the cost measurement.

\subsection{Pietrzak VDF algorithms}

Our research has a model with a smart contract that supports a function that performs Pietrzak VDF verification. Referring to the Pietrzak's original paper \cite{pietrzak2019simple}, we present the evaluation procedure and the halving protocol. In the following subsections, we show how these procedures lead to proof generation and verification in a blockchain environment, culminating in a refined approach that reduces on-chain overhead in the Ethereum Virtual Machine (EVM).

\subsubsection{VDF Evaluation}

The Pietrzak VDF algorithm operates with the following parameters:
\begin{itemize}
    \item $x$: The input value for the VDF evaluation.
    \item $N$: The modulus (an RSA modulus $N = p \cdot q$), where $p$ and $q$ are large primes.
    \item $T$: The time delay parameter, where $T = 2^\tau$ for some $\tau$.
    \item $y$: The output value, computed as $y = x^{2^T} \bmod N$.
\end{itemize}

To evaluate the VDF, one exponentiates $x$ by $2^T \bmod N$, as shown in Equation~\ref{eq: evaluation}:
\begin{equation} \label{eq: evaluation}
    y \equiv x^{2^T} \equiv x^{2^{2^\tau}} \mod{N}.
\end{equation}

The sequential nature of repeated squaring ensures that this operation cannot be significantly parallelized, making it an effective time-delay mechanism. Algorithm~\ref{algorithm: VDF evaluation} details this process, serving as the foundation for proving that the output $y$ can only be computed by incurring a predefined amount of real-time work. This property is crucial for applications requiring verifiable latency and resistance to parallel computation.

\begin{algorithm}[htb]
\caption{Evaluation of the Pietrzak VDF}
\label{algorithm: VDF evaluation}
\begin{algorithmic}[1]
\State \textbf{input:} $x, T, N$
\State \textbf{output:} $y$
\State $y \gets x$
\For{$k \gets 1$ \textbf{to} $T$}
    \State $y \gets y^2 \mod N$
\EndFor
\State \textbf{return} $y$
\end{algorithmic}
\end{algorithm}

\subsubsection{Pietrzak VDF proof generation}

The proof generation defined in Algorithm \ref{algorithm: VDF proof generation} starts with the evaluation result and is used to construct the cryptographic proofs required for VDF verification. The halving protocol detailed in lines 5 -- 13 optimizes the proof generation process for VDFs by reducing the verification time incrementally. At each step, the protocol halves the time parameter $T$, recalculates intermediate values $x_{i}$ and $y_{i}$, and updates these values based on a cryptographic hash function to ensure integrity. As we implement it on Ethereum, we use \texttt{Keccak256} as the cryptographic hash function. This iterative reduction continues until the desired granularity is reached, ensuring that the verifier's workload is minimized while maintaining cryptographic security. Each halving can be tracked as intermediate values $v_i$ are stored as part of the proof.

\begin{algorithm}[htb]
\caption{Pietrzak VDF proof generation}
\label{algorithm: VDF proof generation}
\begin{algorithmic}[1]
\State \textbf{input:} $x, T, N$
\State \textbf{output:} $\{\pi_i\}_{i=1}^{\tau}$
\State $\tau \gets \lfloor \log_2(T) \rfloor$
\State $(x_1, y_1) \gets (x, y)$
\For{$i \gets 1$ \textbf{to} $\tau$}
    \State $v_i \gets x^{2^{T/2^{i-1}}}$
    \State $r_i \gets keccak256(x_i || y_i || v_i)$
    \State $x_{i+1} \gets x_i^{r_i} \cdot v_i \mod N$
    \If{$T/2^{i-1}$ \textbf{is odd}} \State $y_i \gets y_i^2 \mod N$\ \EndIf
    \State $y_{i+1} \gets v_i^{r_i} \cdot y_i \mod N$
    \State $\pi_i \gets v_i$
\EndFor
\State \textbf{return} $\{\pi_i\}_{i=1}^{\tau}$
\end{algorithmic}
\end{algorithm}

\subsubsection{Pietrzak VDF verification}
Algorithm \ref{algorithm: verfication} outlines the Pietrzak VDF verification process. Given that a prover insists the VDF evaluation $x^{2^{T}} = y$ with the time delay parameter $T=2^\tau$, the set of halving proof $\{\pi_i\}_{i=1}^{\tau}$, and the modulus $N$. It repeats the halving protocol (lines 5 -- 13) until the verifier can simply check the result by squaring the base. The halving protocol reduces the number of exponentiations required for verification by half.  The halving protocol outputs $x_1, y_1$ and the proof $\pi$ reasoning $x_1^{2^{2^{\tau-1}}}=y_1$. Therefore, the verifier essentially needs to repeat the halving protocol $\tau$ times for the minimal exponentiation.

One drawback of the Pietrzak VDF is that as the time delay parameter \(T\) increases, the proof size and the number of halving protocol repetitions for the verifier also grow. To address this, Pietrzak introduced a proof-shortening parameter, \(\delta\), which allows the verifier to skip \(\delta\) halving rounds. Specifically, the verifier can repeat the halving protocol \(\tau - \delta\) times instead of \(\tau\) times, then check \(x^{2^{2^\delta}} \stackrel{?}{=} y\). The choice of $\delta$ may be determined by the system designer, the verifier, or specified by protocol parameters, depending on the implementation. This method both reduces the proof size and alleviates verification overhead. Algorithm~\ref{algorithm: refined verfication} describes this refined verification procedure.

In most systems, the costs of each computation and the overhead due to proof size are dependent on the specific system environment, making it difficult to determine the optimal value for the proof shortening parameter, $\delta$. However, in the EVM, the costs for data (calldata) and arithmetic computations are predefined. This led us to hypothesize that it is possible to identify the optimal $\delta$ value for implementation, thereby minimizing both computation and calldata costs in blockchain environments.

\begin{figure}[htb]
    \centering
    \begin{minipage}[t]{0.48\textwidth}
        \begin{algorithm}[H]
        \caption{Pietrzak VDF Verification}
        \label{algorithm: verfication}
        \begin{algorithmic}[1]
        \State \textbf{input:} $x, y, T, \{\pi_i\}_{i=1}^{\tau}, N$
        \State \textbf{output:} True or False
        \State $\tau \gets \lfloor \log_2(T) \rfloor$
        \State $(x_1, y_1) \gets (x, y)$
        \For{$i \gets 1$ \textbf{to} $\tau$}
            \State $v_i \gets \pi_i$
            \State $r_i \gets keccak256(x_i || y_i || v_i)$
            \State $x_{i+1} \gets x_i^{r_i} \cdot v_i \mod N$
            \If{$T/2^{i-1}$ \textbf{is odd}} \State $y_i \gets y_i^2 \mod N$\ \EndIf
            \State $y_{i+1} \gets v_i^{r_i} \cdot y_i \mod N$
        \EndFor
        \If{$y_{\tau+1} = x_{\tau+1}^{2}$}
            \State \textbf{return} True
        \Else
            \State \textbf{return} False
        \EndIf
        \end{algorithmic}
        \end{algorithm}
    \end{minipage}
    \hfill
    \begin{minipage}[t]{0.48\textwidth}
        \begin{algorithm}[H]
        \caption{Refined Pietrzak VDF Verification}
        \label{algorithm: refined verfication}
        \begin{algorithmic}[1]
        \State \textbf{input:} $x, y, T, \delta, \{\pi_i\}_{i=1}^{\tau-\delta}, N$
        \State \textbf{output:} True or False
        \State $\tau \gets \lfloor \log_2(T) \rfloor$
        \State $(x_1, y_1) \gets (x, y)$
        \For{$i \gets 1$ \textbf{to} $\tau-\delta$}
            \State $v_i \gets \pi_i$
            \State $r_i \gets keccak256(x_i || y_i || v_i)$
            \State $x_{i+1} \gets x_i^{r_i} \cdot v_i \mod N$
            \If{$T/2^{i-1}$ \textbf{is odd}} \State $y_i \gets y_i^2 \mod N$\ \EndIf
            \State $y_{i+1} \gets v_i^{r_i} \cdot y_i \mod N$
        \EndFor
        \If{$y_{\tau-\delta+1} = x_{\tau-\delta+1}^{2^{2^\delta}}$}
            \State \textbf{return} True
        \Else
            \State \textbf{return} False
        \EndIf
        \end{algorithmic}
        \end{algorithm}
    \end{minipage}
    \caption{Algorithms for Pietrzak VDF Verification and Refined Verification.}
    \label{fig:combined-algorithms}
\end{figure}

\subsection{Target Bit Length and Target Delay}

The security of Pietrzak VDF configuration depends on the difficulty of the prime factoring problem, akin to that in the RSA encryption system. Currently, a 2048-bit key length is deemed sufficient for maintaining security standards. However, considering potential advancements in computational capabilities, a 3072-bit key length might be necessary after ten years. Therefore, our experiments included tests with both 2048-bit and 3072-bit settings for VDFs.

In a recent study on VDF implementation using C++ \cite{attias2020implementation}, VDF evaluation times with the time delay parameter \(T = 2^\tau\) were examined. The study found that for \(\tau\) values ranging from 20 to 25, the evaluation can be completed within one minute using a personal CPU. The evaluation involves performing \(y = x^{2^T}\). We chose to conduct our VDF experiments within the same \(\tau\) range, as it strikes a balance between computational feasibility and practical utility while avoiding excessive resource consumption.

\subsection{Environment for Gas Usage Measurement}
\label{subsection:experiment-environment}

\subsubsection{Development Environment and Tools}

The development of the smart contract was conducted using Solidity version 0.8.26 and EVM version Cancun. At the start of the development process, it was the latest compiler version fully supported by Foundry, which is a comprehensive Ethereum development toolchain that provides a robust environment for building, testing, and debugging smart contracts. Furthermore, gas optimization is best achieved with the latest version of the compiler.

\subsubsection{Optimizer Configuration}

The core parameter for the Solidity optimizer configuration is \texttt{Runs} which indicates how often each opcode will be used in the deployed contract. That is, higher \texttt{Runs} implies less gas cost for execution and higher gas cost for deployment. We set the optimizer parameter \texttt{Runs} to 4,294,967,295($2^{32} - 1$), which is the maximum value of the parameter, indicating the highest level of optimization for execution.

We set up viaIR (Intermediate Representation) and the Yul optimizer in the Solidity compiler to generate efficient bytecode. viaIR helps efficiency, security, and simplicity in code generation by deviating from the direct compilation of Solidity code to EVM bytecode. Instead, viaIR processes an intermediate step where Yul code, an intermediate representation, is generated. The Yul optimizer refines this intermediate representation, enhancing the IR produced by the pipeline, as well as inline assembly and utility Yul code generated by the compiler \cite{soliditycompilerdocs}.

The \textsf{msize()} operation is incompatible with the Yul optimizer, necessitating its replacement with the free memory pointer. This adjustment required manual updates to memory indexes, which were rigorously tested to ensure the integrity of our implementation. To comply with Solidity's memory model and enable viaIR, we incorporated memory-safe assembly blocks. These blocks ensure robust memory management and adhere to Solidity's safety constraints, allowing the viaIR pipeline to be activated, resulting in significant gas savings \cite{soliditymemorysafetydocs}.

\subsubsection{Gas Measurement in Local Testing Environment}

As the test target code, we implemented contracts for each test in one main source file ensuring that each contract contains only one external function. This approach minimizes the error margin that can arise when multiple external functions are present, as the order of the function selectors impacts the \textsf{gasUsed}.

For the cost (gas usage) measurement, we utilized the Anvil node, provided by the Foundry toolchain, to create a local testing environment. The gas usage was obtained using the Forge testing framework provided by Foundry. The Forge standard library's \textsf{Test} contract includes the \textsf{lastCallGas()} function, which retrieves the \textsf{gasUsed} in the last call. Additionally, \texttt{isolate} flag must be set when running the Forge's test command to ensure all top-level calls are executed as separate transactions in distinct EVM contexts. This setup allows for more precise gas accounting. It is important to note that the actual gas cost incurred on the real network is determined by multiplying the estimated \textsf{gasUsed} by the \textsf{gasPrice} at the Ethereum network circumstance.

\section{Cost measurement and approximation} \label{section: measure}

In this section, we introduce the data structure we used to send VDF proofs in the Ethereum transaction calldata. Then, we investigate the gas cost to send VDF proofs and regress the results to linear functions for the next sections.

\subsection{Data Structure}

In this study, the VDFs are implemented with key sizes of 2048 and 3072 bits. However, the EVM natively supports arithmetic operators only up to 256-bit numbers. To facilitate efficient computation of numbers exceeding 256 bits, it is necessary to employ customized data formats. We revised the big number library implementation for Solidity developed by Firo \cite{firoorg_solidity_BigNumber}. The original data structure in this library includes bytes data (\texttt{bytes} val), bit length (\texttt{uint} bitlen), and a sign indicator (\texttt{bool} neg). Among the three elements, we deleted the sign indicator and its related low-level logic since our implementation does not require handling negative values. This library optimizes on-chain operations such as addition and subtraction by adding leading zeros for 32-byte alignment with the EVM's memory word size, eliminating offset management. Also, it utilizes the bit length data field to optimize multiplication and comparison operations. Therefore, we format the bytes data and include the bit length alongside it, significantly reducing unnecessary on-chain computations. For example, the data format for a 2048-bit representation in JSON is as shown in Listing \ref{lst:example}. In this way, we can reduce the gas cost even though the calldata size increases. This decision is reasonable because the primary goal is to reduce the gas cost, not the transaction calldata size.

\begin{lstfloat}[htb]
\begin{lstlisting}[language=json]
{
  "big_number": {
    "val": "0x4621c26320fe0924bba1b7d5bc863495b9f0db3823b12a9a18e21d23...........b2127fd6390f5f6234a164bd39d1dd6884b768c4dd790586ee",
    "bitlen": 2047
  }
}
\end{lstlisting}
\caption{Calldata for a Big Number Example}
\label{lst:example}
\end{lstfloat}

\subsection{Gas Cost for Calldata and Dispatch}

The gas cost for Ethereum transaction calldata, as described in EIP-2028 "Transaction data gas cost reduction" \cite{eip2028}, is divided into two cases: 16 gas per non-zero byte and 4 gas per zero byte.

Considering the above gas pricing, we can calculate the cost for sending one 2048-bit length data in the aforementioned data structure. Taking into account the probability of each byte being non-zero, which is \(\frac{1}{256}\), the expected gas cost for sending one 2048-bit length data is 1036 gas. The gas for the Big Number structure is added to this.

In our implementation, the function \texttt{verifyRecursiveHalvingProof} processes the refined verification, Algorithm \ref{algorithm: refined verfication}. 
A call to this function involves encoding parameter values. The process starts by computing the \texttt{keccak256} of the UTF-8 byte representation of the string\footnote{verifyRecursiveHalvingProof((bytes,uint256)[], (bytes,uint256), (bytes,uint256), (bytes,uint256), uint256, uint256)}, then using the first four bytes, which results in \texttt{0xd8e6ac60}.

For static types such as \texttt{uint256}, fields including \texttt{delta} and \texttt{T} are padded to 32 bytes. On the other hand, In EVM, dynamic types, including \texttt{bytes}, \texttt{struct}, and dynamic arrays, are encoded by specifying an offset to the start of the data, its length, and then the data itself. This routine is applied recursively for nested dynamic types. In case of \texttt{BigNumber[] memory v}, it is the dynamic array of dynamic type struct \texttt{BigNumber} that contains dynamic bytes \texttt{val}. So the way dynamic types are encoded is recursively applied to each halving proof \texttt{v}.

Parameters such as $T$, $\delta$, $\lambda$, and parameter count (e.g., 6 in our model) determine the offsets and lengths. The gas usage for data depends probabilistically on the distribution of zero and non-zero bytes in the data. Since EVM organizes memory in 32-byte words, we can consider the same number of \texttt{word} is repeated when data is repeated. Finally, we can predict the approximate gas usage of calldata using this EVM data construction structure.

To estimate the gas usage for VDF proof calldata, we predicted the data composition and its associated costs by analyzing the probability distribution of zero and non-zero bytes in the serialized data. This prediction was derived using the EVM's gas pricing model for calldata, where zero bytes incur 4 gas units and non-zero bytes incur 16 gas units as defined in EIP-2028 \cite{eip2028}. Based on these predictions, we conducted experiments to measure the actual calldata size and the corresponding gas usage under various proof configurations.

The related gas usage comprises of the intrinsic gas and data dispatching gas. The intrinsic gas in EVM refers the minimum cost for a transaction. It contains the constant transaction cost (21000 gas) and the cost for data supply. For verification, we need not only to allocate data in a transaction, but also to dispatch data into the verification function. The dispatching cost is proportional to the data size. Therefore, we decided to include the dispatching cost into the data cost in this study.
Figure \ref{fig:calldata-2048} for $\lambda=2048$ and \ref{fig:calldata-3072} for $\lambda=3072$ show the calldata size with the number of proofs and the related gas usage. We can see the gas usage and the calldata size simply and linearly increase with the number of proofs.

\begin{figure}[!tp]
    \centering
    \begin{subfigure}[t]{0.48\linewidth}
        \includegraphics[width=\linewidth]{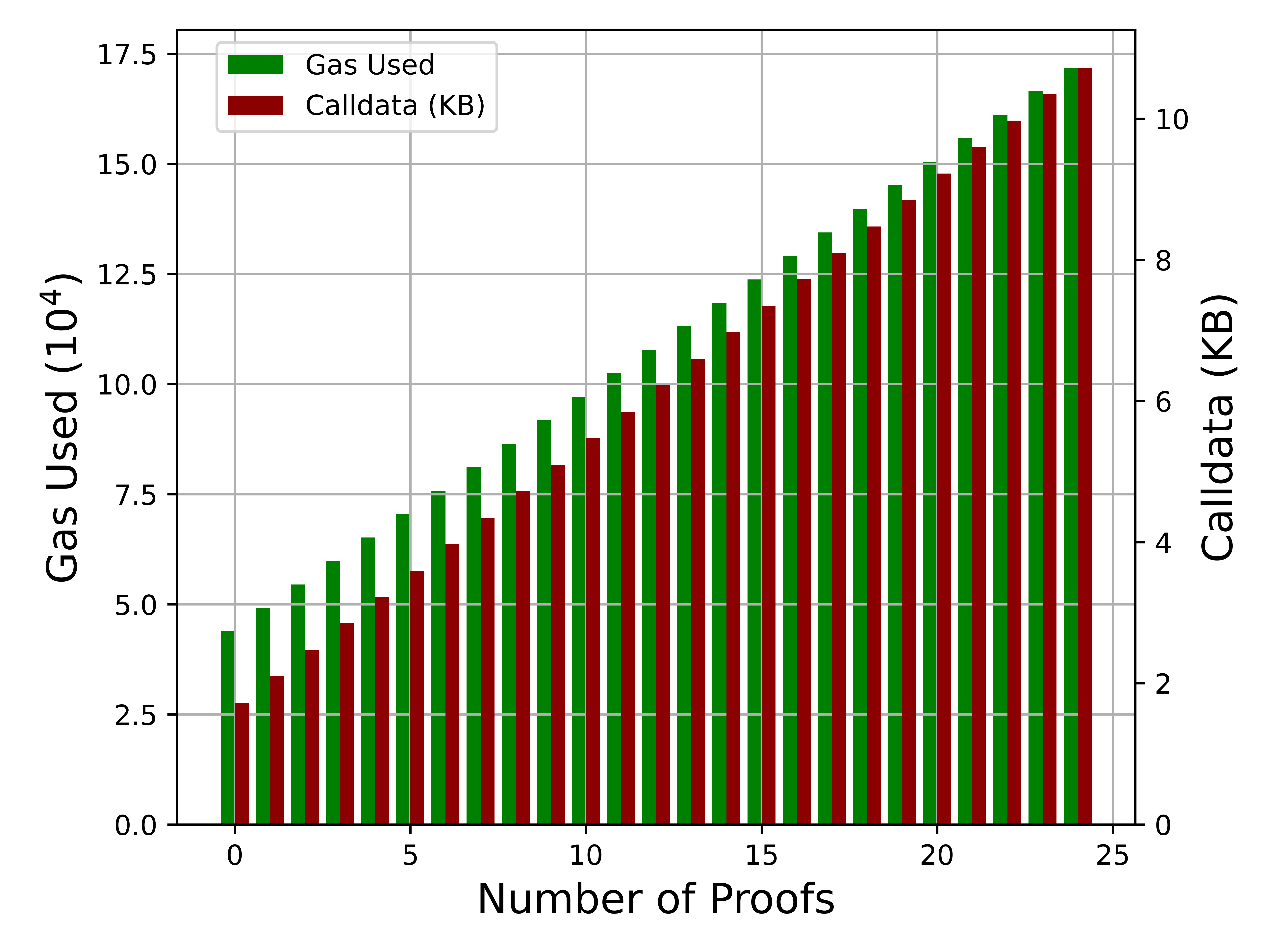}
        \caption{Calldata Cost, $\lambda=2048$}
        \label{fig:calldata-2048}
    \end{subfigure}
    \hfill
    \begin{subfigure}[t]{0.48\linewidth}
        \includegraphics[width=\linewidth]{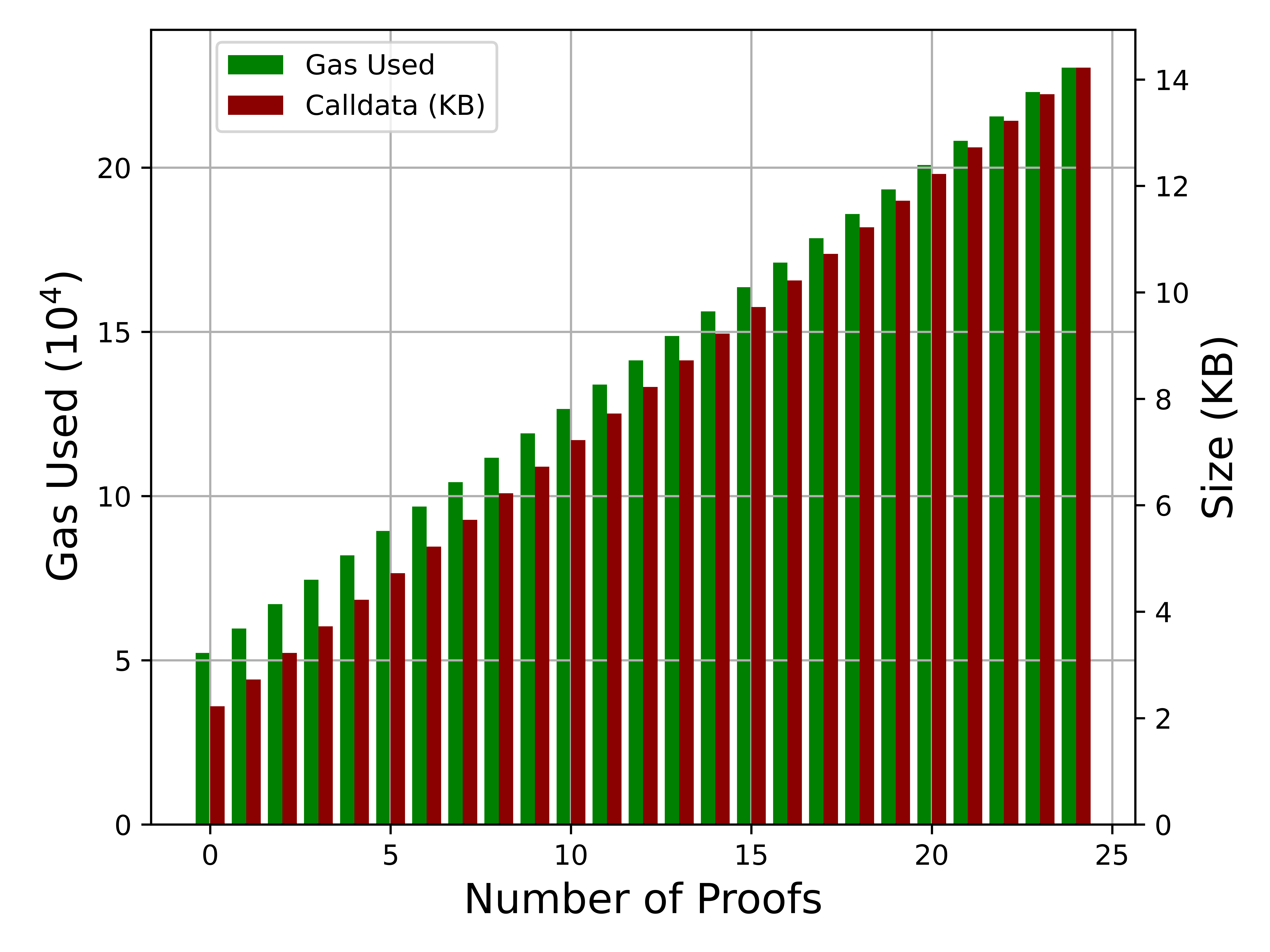}
        \caption{Calldata Cost, $\lambda=3072$}
        \label{fig:calldata-3072}
    \end{subfigure}

    \vspace{0.5cm} %
    \begin{subfigure}[t]{0.48\linewidth}
        \includegraphics[width=\linewidth]{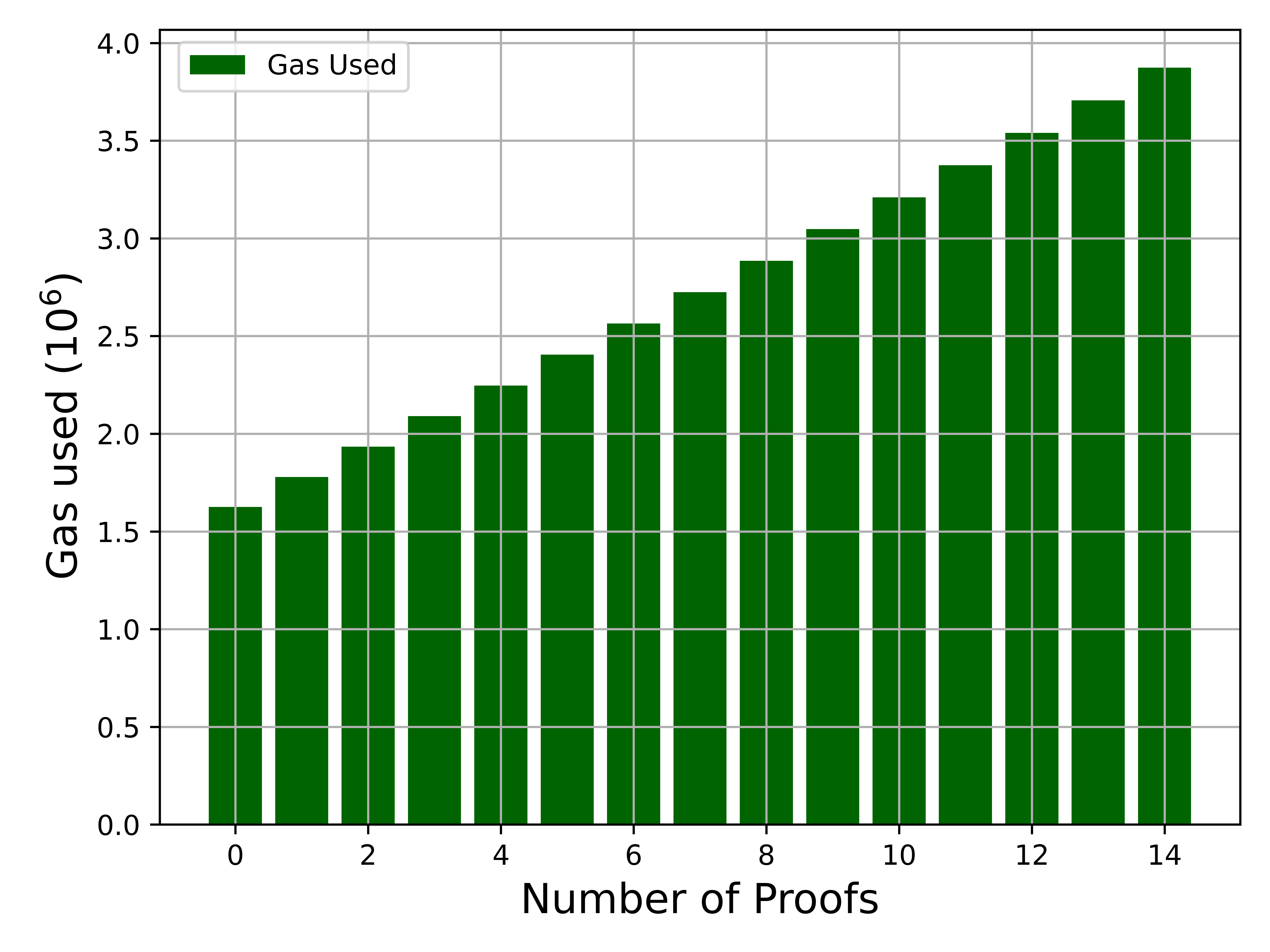}
        \caption{Halving Cost, $\lambda=2048$}
        \label{fig:halving-2048}
    \end{subfigure}
    \hfill
    \begin{subfigure}[t]{0.48\linewidth}
        \includegraphics[width=\linewidth]{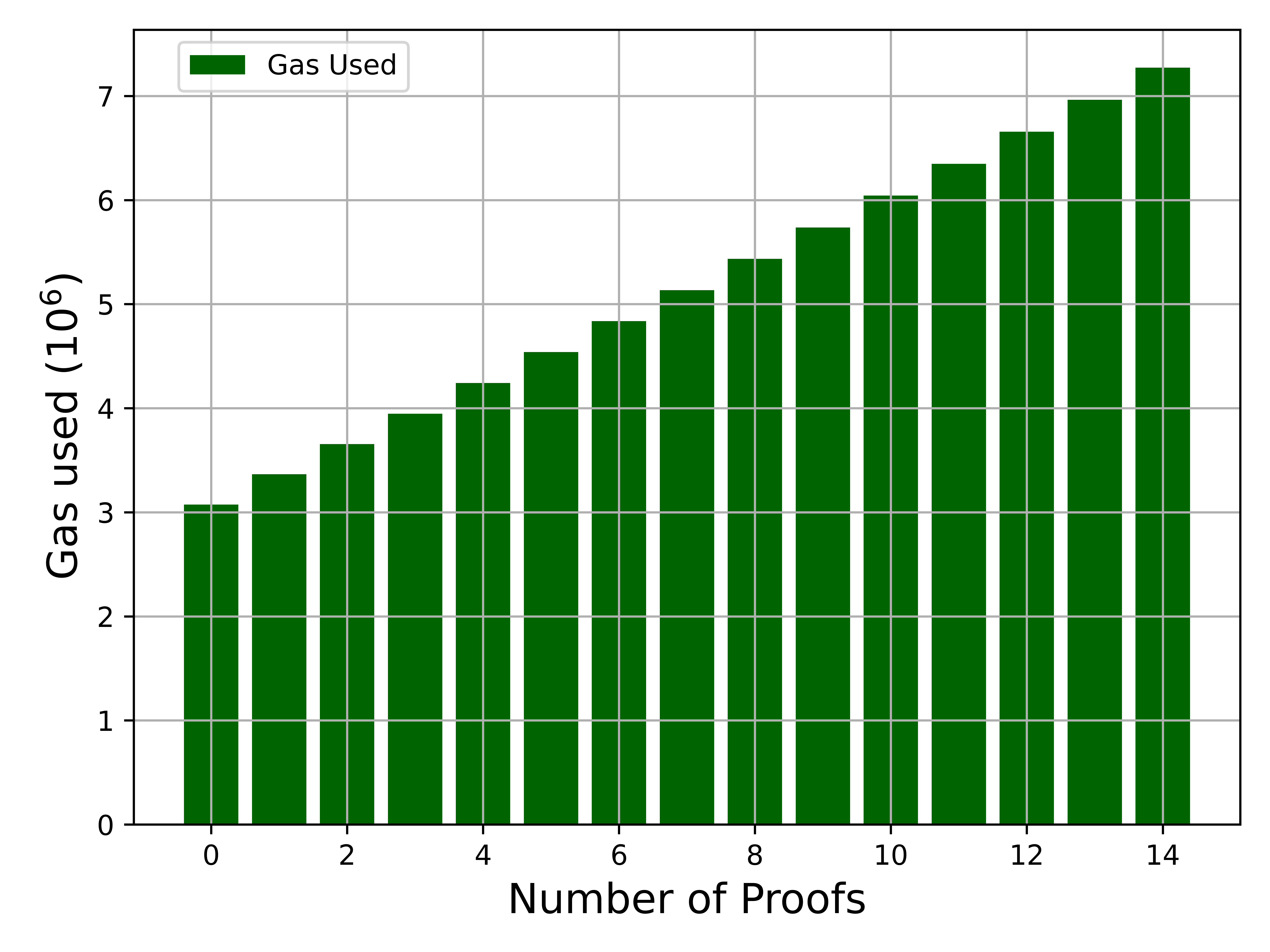}
        \caption{Halving Cost, $\lambda=3072$}
        \label{fig:halving-3072}
    \end{subfigure}

    \vspace{0.5cm} %
    \begin{subfigure}[t]{0.48\linewidth}
        \includegraphics[width=\linewidth]{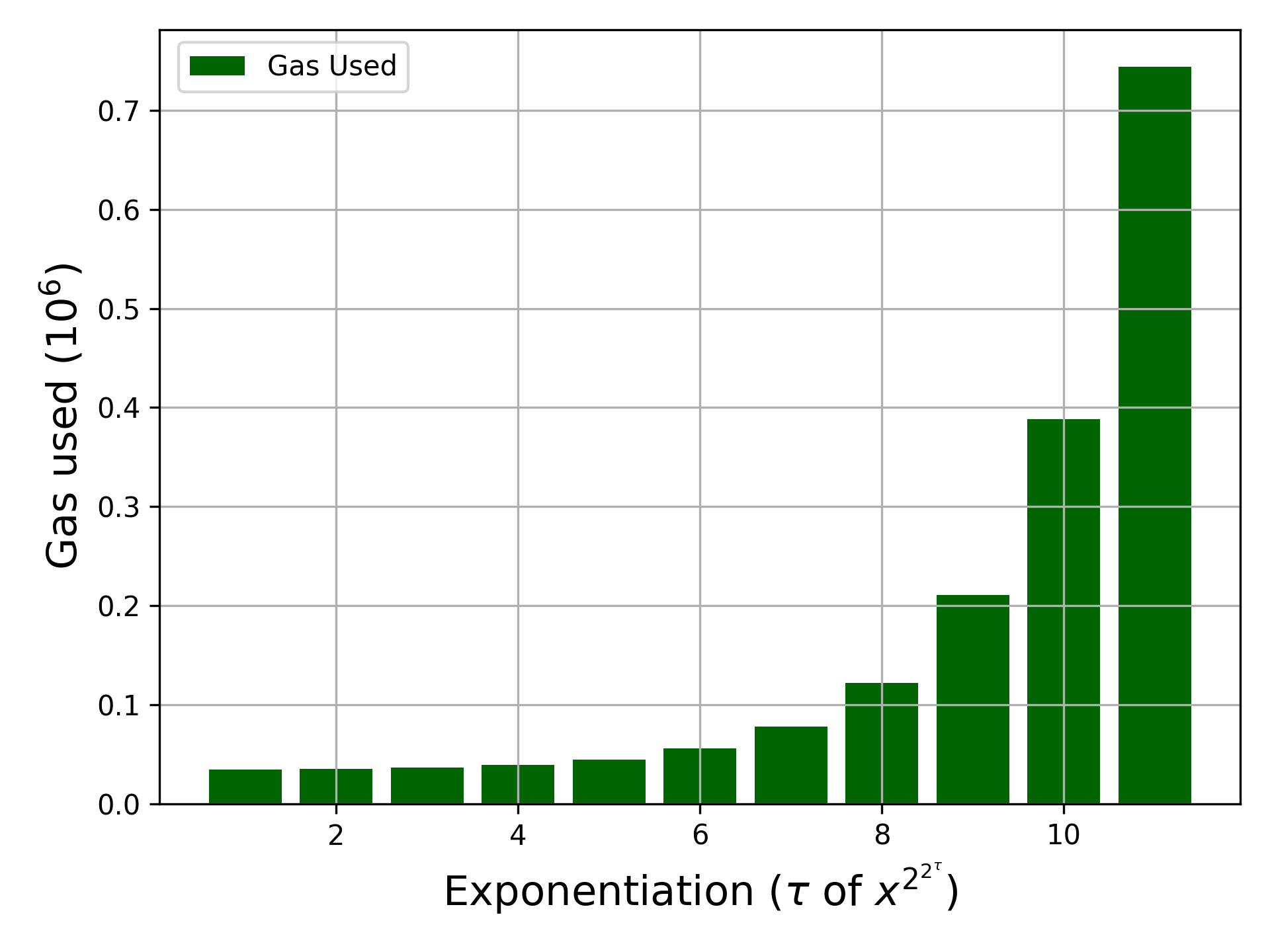}
        \caption{\texttt{ModExp} Cost, $\lambda=2048$}
        \label{fig:modexp-2048}
    \end{subfigure}
    \hfill
    \begin{subfigure}[t]{0.48\linewidth}
        \includegraphics[width=\linewidth]{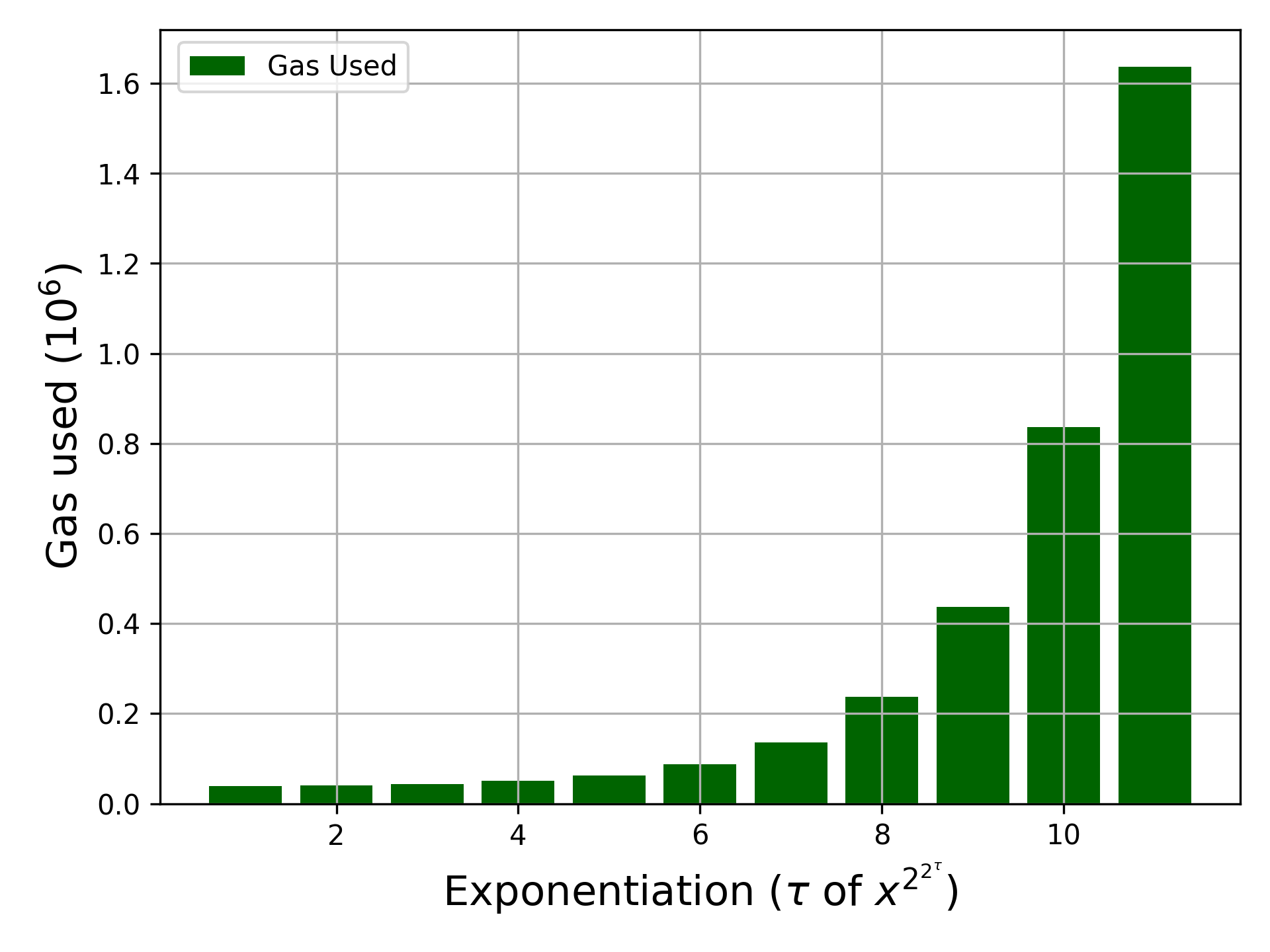}
        \caption{\texttt{ModExp} Cost, $\lambda=3072$}
        \label{fig:modexp-3072}
    \end{subfigure}

    \caption{Experiment results: calldata cost, halving cost, and \texttt{ModExp} cost for different configurations ($\lambda=2048$ and $\lambda=3072$).}
    \label{fig:combined-experiment-results}
\end{figure}

Finally, to describe the relationship between transaction size and gas costs, we performed a basic regression analysis. By applying the least squares method \cite{montgomery2021introduction}, which finds the best-fit line by minimizing the sum of the squared differences between observed and predicted values, we derive the following linear regression equations \ref{eq:data-2048} and \ref{eq:data-3072} for data sizes of 2048 and 3072 bytes, respectively:

\begin{equation} \label{eq:data-2048}
    \mathcal{G}_{data, 2048}(\tau) = 5295.73 \cdot \tau + 37610.11
\end{equation}

\begin{equation} \label{eq:data-3072}
    \mathcal{G}_{data, 3072}(\tau) = 7389.31 \cdot \tau + 43943.88
\end{equation}

These equations estimate the gas cost $\mathcal{G}$ as a function of transaction size $\tau$, highlighting the increase in required gas with larger transaction calldata.

\subsection{Modular Exponentiation Cost}

EIP-198 \cite{eip198}, known for "Big Integer Modular Exponentiation", introduced the \texttt{ModExp} (Modular Exponential) precompiled contract to the EVM, offering a streamlined and cost-efficient approach for large integer modular exponentiation. \footnote{EIP-2565 \cite{eip2565} contributed to further refining the gas cost model for these modular exponentiation operations by proposing adjustments to the pricing formula.} 
As its obvious that a precompile provides the best efficinecy, we applied \texttt{ModExp} for both the modular exponentiation and multi-exponentiation operations in the implementation in Algorithm \ref{algorithm: refined verfication}.
In Figure \ref{fig:modexp-2048} and Figure \ref{fig:modexp-3072} illustrating the experiment results, the pricing of \texttt{ModExp} is linearly proportional to the number of exponentiation as $ \mathcal{G} (T) = aT + b$, where we denote $\tau$ as the logarithmic value of $T$. We could get the following results from a simple regression process.

\begin{equation} \label{eq:exp-2048}
    \mathcal{G}_{exp, 2048} (\tau) = 346.92 \cdot 2^\tau + 33432.08
\end{equation}

\begin{equation} \label{eq:exp-3072}
    \mathcal{G}_{exp, 3072} (\tau) = 780.83 \cdot 2^\tau + 37445.60
\end{equation}

\subsection{Halving Cost}

The halving cost of the arithmetic operations in the verification indicates the gas usage for lines 5 -- 13 in Algorithm \ref{algorithm: refined verfication}. The modular exponentiation operations in lines 8, 10, and 12 use the precompiled \texttt{ModExp}. For simplicity, as we assumed $T = 2^\tau$, in this case, the if statement in line 9 is always false, therefore line 10 is never executed. Hence, the gas usage is clearly proportional to the number of the repetition time, that is, $\tau-\delta$.

The results of our experiments demonstrate the anticipated correlation. Figure \ref{fig:halving-2048} and Figure \ref{fig:halving-3072} display the gas cost of the repeated halving operations with $\lambda=2048$ and $\lambda=3072$ correspondingly. The data exhibits a clear linear relationship, allowing us to perform a regression analysis and model the results using a linear function. By employing the method of least squares approximation, we obtain the subsequent functions as Equation \ref{eq:halving-2048} and \ref{eq:halving-3072}:

\begin{equation} \label{eq:halving-2048}
    \mathcal{G}_{halving, 2048}(\tau) = 160604.41 \cdot \tau - 157178.66
\end{equation}

\begin{equation} \label{eq:halving-3072}
    \mathcal{G}_{halving, 3072}(x) = 299896.98 \cdot \tau - 248119.67
\end{equation}

\section{Cost-Effective Implementation} \label{section: result}

This section explains the implementation of a cost-efficient Pietrzak VDF verifier as an Ethereum smart contract. In the beginning, we utilize the gas cost outcomes obtained from the preceding sections to make theoretical predictions for the optimized proof generating process. Secondly, we prove that there exists a unique \(\delta\) that minimizes the gas usage for VDF verification, and this minimization is independent of \(\tau\). In the end, we compare the anticipated outcome based on theory with the actual results obtained from the experiment.

\subsection{Theoretical Analysis on Proof Generation for Gas Usage Optimization}

The optimized proof generation in this research means minimizing the verification gas cost in EVM. We can define the total gas cost function $g_{total}$ using the regression results (Eq. \ref{eq:data-2048} -- \ref{eq:halving-3072}) in the previous section.

\begin{equation} \label{eq:total gas}
    \mathcal{G}_{\text{total}}(\tau, \delta) = \mathcal{G}_{\text{calldata}} + \mathcal{G}_{\text{halving}} + \mathcal{G}_{\text{exp}} + C
\end{equation} 

\begin{equation}
    = \big( \underbrace{\alpha \cdot (\tau-\delta) + c_1}_{\mathcal{G}_{\text{calldata}}} \big) 
    + \big( \underbrace{\beta \cdot (\tau-\delta) + c_2}_{\mathcal{G}_{\text{halving}}} \big) 
    + \big( \underbrace{\gamma \cdot 2^\delta + c_3}_{\mathcal{G}_{\text{exp}}} \big) + C
\end{equation}

\begin{equation}
    = (\alpha+\beta)(\tau-\delta) + \gamma \cdot 2^\delta + C'
\end{equation}

Then, we get Theorem~\ref{theorem: minimizing delta} says that there is a unique delta minimizing the total gas cost. \newline

\begin{theorem} \label{theorem: minimizing delta}
The function \(\mathcal{G}_{total}(\tau, \delta)\) has exactly one unique minimum for \(\delta\) when \(\tau > \delta = \log_2 (\alpha + \beta) - \log_2 (\ln{2} \cdot \gamma) \).
\end{theorem}

\begin{proof}
    See \ref{section: proof of minimum delta}.
\end{proof}

Also, we denote the minimizing \(\delta\) as

\begin{equation}
    \delta_m = \arg\min_{\delta} \mathcal{G}_{total}(\tau, \delta).
\end{equation}

From Theorem \ref{theorem: minimizing delta}, we can get the unique \(\delta_m\) optimizing the gas cost. Applying all the regression results in Equation \ref{eq:data-2048}, \ref{eq:data-3072}, \ref{eq:exp-2048}, \ref{eq:exp-3072}, \ref{eq:halving-2048}, and \ref{eq:halving-3072}, we get the optimizing  \(\delta_{m} \approx 9.43 \) for \(\lambda = 2048 \) and \(\delta_{m} \approx 9.15 \) for \(\lambda = 3072 \) respectively. In practice, \(\delta\) for implementation must be an integer. Therefore, we can get the practical integer \(\delta\) with the following corollary. \newline

\begin{corollary} \label{corollary:integer delta}
Let \(\delta_m\) be a real number, and let \(\mathcal{G}(\tau, \delta)\) denote a function in \(\tau\) and \(\delta\). Define \(\Delta = \delta_m - \lfloor \delta_m \rfloor\) as the fractional part of \(\delta_m\). Then, the relationship between \(\mathcal{G}(\tau, \lfloor \delta_m \rfloor)\) and \(\mathcal{G}(\tau, \lceil \delta_m \rceil)\) depends on \(\Delta\) as follows:
\begin{enumerate}
    \item If \(\Delta < -\log_2(\ln 2)\), then \(\mathcal{G}(\tau, \lfloor \delta_m \rfloor) < \mathcal{G}(\tau, \lceil \delta_m \rceil)\).
    \item If \(\Delta = -\log_2(\ln 2)\), then \(\mathcal{G}(\tau, \lfloor \delta_m \rfloor) = \mathcal{G}(\tau, \lceil \delta_m \rceil)\).
    \item If \(\Delta > -\log_2(\ln 2)\), then \(\mathcal{G}(\tau, \lfloor \delta_m \rfloor) > \mathcal{G}(\tau, \lceil \delta_m \rceil)\).
\end{enumerate}
\end{corollary}

\begin{proof}
    See \ref{section: proof of integer delta}.
\end{proof}

Since \(\log_2(\ln 2) \approx -0.52\), the optimal integer value of \(\delta\) is  9 for both \(\lambda = 2048\) and \(3072\), as the 1st case of Corollary \ref{corollary:integer delta}.

\subsection{Optimized Verification Gas Cost}

Figures \ref{fig:delta-2048} and \ref{fig:delta-3072} depict the total gas cost for Pietrzak VDF verification at \(\lambda=2048\) and \(\lambda=3072\) respectively. Consistent with the theoretical predictions in Theorem \ref{theorem: minimizing delta}, the gas usage graph shows a decline to a minimal point before increasing sharply. Optimal gas costs were observed at \(\delta = 9\) for both configurations, aligning perfectly with the findings of Corollary \ref{corollary:integer delta}. This alignment between theoretical expectations and empirical results confirms the accuracy of our model and highlights its potential for optimizing Pietrzak VDF verification in blockchain applications.

Table \ref{table:final-results} details the gas consumption and \texttt{calldata} sizes for the implemented parameters, uniformly using \(\delta = 9\). Contrary to the previous expectation of \texttt{calldata} sizes around 40KB \cite{attias2020implementation}, the sizes recorded were significantly lower, ranging from 5.38 to 7.25 KB for \(\lambda = 2048\) and from 7.13 to 9.63 KB for \(\lambda = 3072\). Gas usage varied from 1.97 to 2.80 million for \(\lambda = 2048\) and from 3.68 to 5.25 million for \(\lambda = 3072\)\footnote{Assuming 1 ETH is valued at \$2600, the cost for \(\lambda = 2048\) ranges from approximately \$33.9 to \$136.6 USD when \(\tau = 25\) and \texttt{gasPrice} varies between 5 and 20 gwei.}.  These findings not only demonstrate the efficacy of our approach but also highlight potential areas for further optimization in VDF implementations. Such improvements could lead to significant cost reductions in blockchain networks, enhancing the overall efficiency of cryptographic operations.

\begin{figure*}[!htp]
    \centering
    \caption{Results for halving proof skipping parameter \(\delta\) with \(\tau = 20, 21, 22, 23, 24, 25\) for $y=x^{2^{2^{\tau}}}$}
    \label{figure:delta experiment}
    \begin{subfigure}[t]{0.47\linewidth}
        \includegraphics[width=\linewidth]{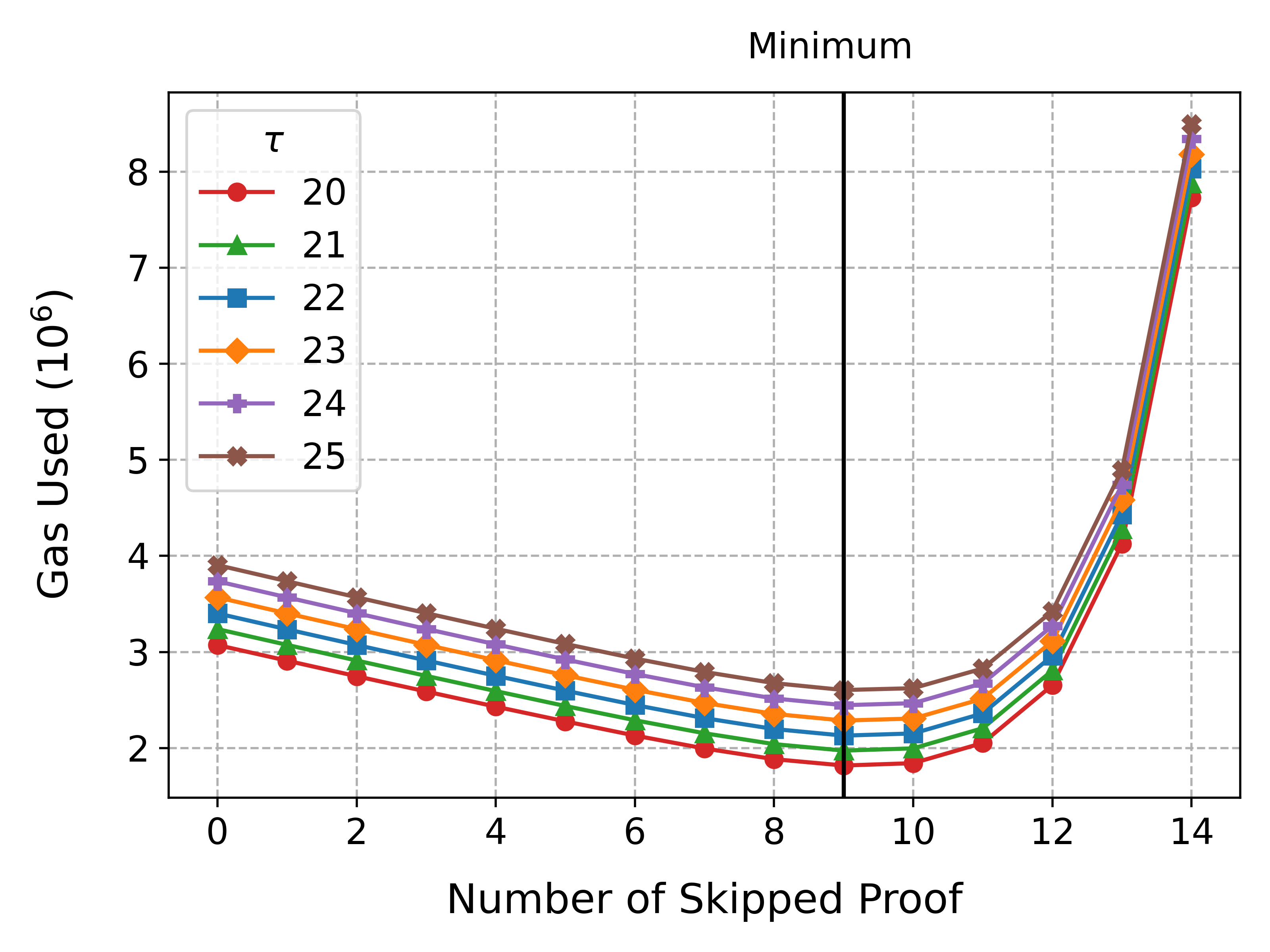}
        \caption{Skipping parts of halving proofs ($\lambda=2048$)}
        \label{fig:delta-2048}
    \end{subfigure}
    \begin{subfigure}[t]{0.47\linewidth}
        \includegraphics[width=\linewidth]{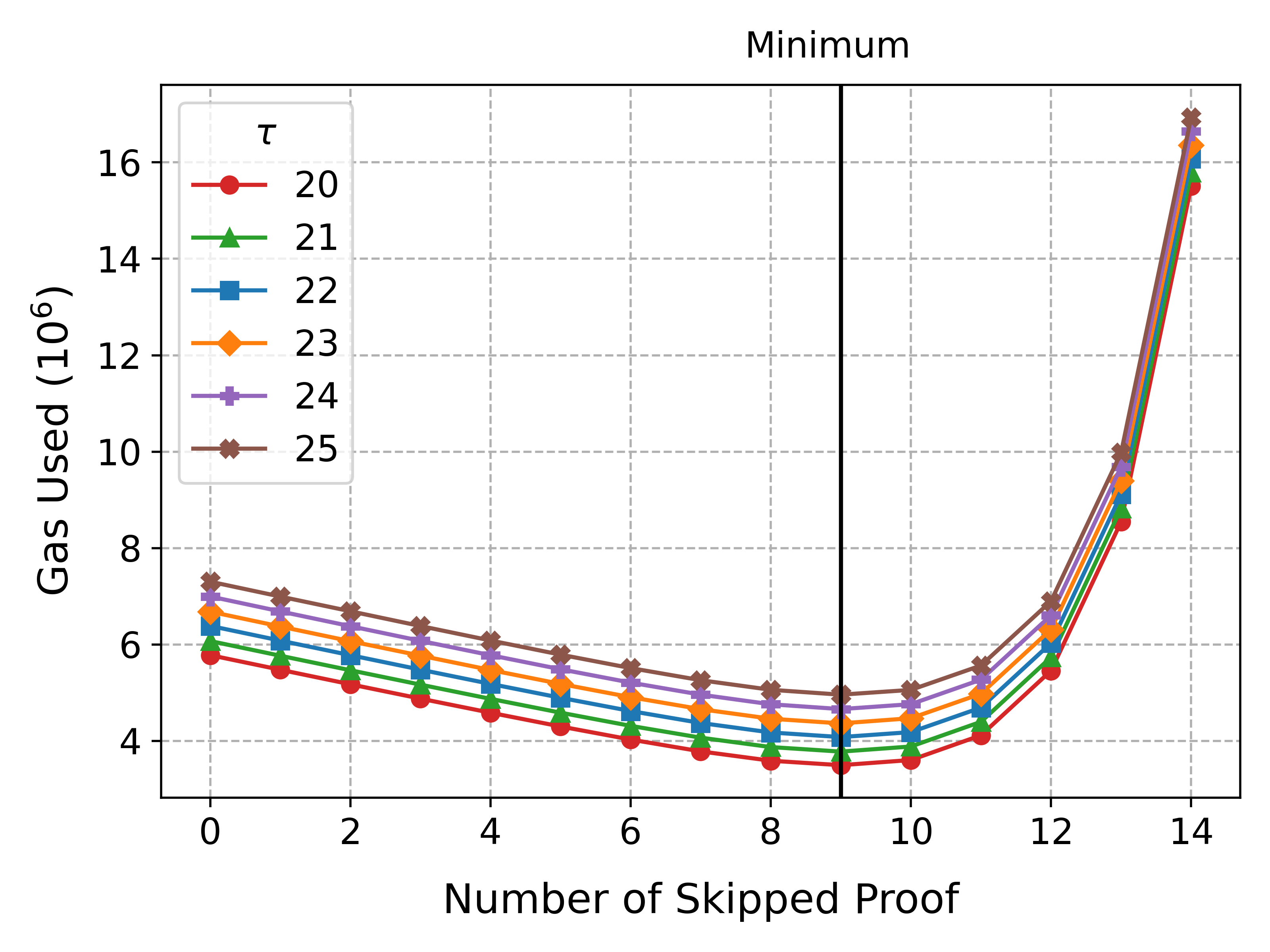}
        \caption{Skipping parts of halving proofs ($\lambda=3072$)}
        \label{fig:delta-3072}
    \end{subfigure}
\end{figure*}

\begin{table}[htb]
\centering
\caption{Implementation results based on the analysis}
\label{table:final-results}
\begin{tabularx}{0.9\textwidth}{x{1.5cm}x{1.5cm}x{1.5cm}*{2}{C}}
\toprule
$\lambda$                & $\delta$              & $\tau$ & Gas used & Calldata size (KB) \\ \midrule
\multirow{5}{*}{2048} & \multirow{5}{*}{9} & 20  &  1817855 &   5.38    \\ \cmidrule{3-5} 
                      &                    & 21  &  1973222 &   5.75        \\ \cmidrule{3-5} 
                      &                    & 22  &  2129116 &   6.13        \\ \cmidrule{3-5} 
                      &                    & 23  &  2286058 &   6.50        \\ \cmidrule{3-5} 
                      &                    & 24  &  2445803 &   6.88       \\ \cmidrule{3-5} 
                      &                    & 25  &  2603753 &   7.25       \\ \midrule
\multirow{5}{*}{3072} & \multirow{5}{*}{9} & 20  & 3495396  &   7.13       \\ \cmidrule{3-5} 
                      &                    & 21  & 3776893  &   7.63        \\ \cmidrule{3-5} 
                      &                    & 22  & 4078468  &   8.13        \\ \cmidrule{3-5} 
                      &                    & 23  &  4362630  &   8.63        \\ \cmidrule{3-5} 
                      &                    & 24  &  4658937  &   9.13        \\ \cmidrule{3-5} 
                      &                    & 25  &  4958138 &   9.63        \\ \bottomrule
\end{tabularx}
\end{table}

\section{Discussion}
\label{section: discussion}

This section examines the broader implications and feasibility of implementing Pietrzak VDF verification on Ethereum. First, we consider the feasibility of the observed gas costs and discuss how frequent verification might remain economically viable on both Ethereum’s Layer~1 and Layer~2. Next, we address the lack of an established economic model for VDFs and the challenges it poses for incentivizing participants in decentralized environments. Lastly, we look beyond the EVM and outline how the findings could extend to more general computing contexts, suggesting avenues for future research and optimization.

\subsection{Practical Feasibility of Our Observed Costs.}
From Section~\ref{section: result}, we note that even with an optimal parameter of $\delta = 9$, the gas cost for verifying a Pietrzak VDF can reach a few million gas units. On Ethereum Layer 1, this is still somewhat expensive but may be acceptable for security-critical operations that are not executed frequently (e.g., infrequent random-beacon generation or governance steps). However, in Layer 2 environments with significantly reduced gas prices, these same verification procedures can become much more practical and cost-effective, especially for repeated or more frequent VDF verifications. By striking this balance—an optimal $\delta$ and a cheaper execution environment—our approach could feasibly support advanced use cases on EVM-compatible rollups, offering robust cryptographic guarantees at a fraction of Layer 1 costs.

\subsection{Lack of Economic VDF Protocols in Ethereum}

Integrating VDFs into smart contracts offers significant potential to enhance blockchain security and reliability, yet faces major challenges due to the lack of a robust economic model tailored to their unique operational needs, such as incentivizing both provers and verifiers and managing commit-reveal-recover protocols. This absence hampers practical deployment, as it fails to compensate those who contribute significant computational resources, deterring potential contributors and affecting scalability and effectiveness in real-world applications. Additionally, the decentralized nature of blockchain necessitates economic strategies that support technical operations and ensure sustainability without central governance, promoting honest participation and deterring dishonesty to maintain integrity and fairness within the blockchain ecosystem. Thus, innovative economic models specifically designed for decentralized environments are crucial to make VDFs both technically effective and economically viable.

\subsection{Future Works on a General Environment}

Through this research, we investigated the efficient implementation of Pietrzak VDF verification within the Ethereum Virtual Machine (EVM) environment and observed results that align both theoretically and experimentally. During our study, we identified significant effects of generalized proof generation and verification, which were also included within the predictable realms of theoretical analysis. We anticipate that such an approach could be applicable not only within the EVM but also in more general computing environments. Although resource pricing in general systems would differ from the gas costs in the EVM, with elements like the repetition of halving protocols and resource consumption proportional to the proof length being predictable, the costs of exponentiation in VDF (\(2^{2^T}\)) are expected to be exponential. Therefore, based on the outcomes of this study, further research on optimizing Pietrzak VDFs in general computing environments appears reasonable.

\section{Conclusion and Future Works} \label{section: conclusion}

Our study demonstrates that Pietrzak's VDF can be effectively optimized for the EVM, reducing gas costs from 4M to 2M and minimizing proof size to under 8 KB with a 2048-bit RSA key. These optimizations validate the theoretical models and show practical viability. While our paper does not address the Wesolowski VDF implementation, which poses challenges due to its hash-to-prime function, exploring indirect or approximate implementation methods could make it feasible for the EVM. Future work should focus on utilizing Wesolowski VDF verification on EVM as it has a big benefit on the constant and short proof length and verification cost. We plan to implement on-chain Wesolowski VDF verification both directly and indirectly.
For indirect implementation, we suggest executing the hash-to-prime function in an off-chain, challenge-based verification system (inspired by Truebit\footnote{Truebit \cite{teutsch2024scalable} is a protocol that outsources computation off-chain while allowing anyone to challenge an allegedly correct result within a specific dispute window. If the challenge proves the result incorrect, a fraud proof is submitted on-chain, ensuring only correct computations are ultimately accepted.}). This approach enables a “proof-of-fraud” mechanism, where a claimed correct computation can be disputed during a challenge period, thereby minimizing on-chain costs without compromising correctness.

\bibliographystyle{ieeetr}
\bibliography{reference}  

\appendix

\section{Code availability} \label{appendix: code availability}

To facilitate reproducibility and peer review, we have made the source code for our implementation available anonymously. The repository contains all the scripts and supplementary materials necessary to reproduce the results discussed in this paper. You can access the repository at the following link: \url{https://github.com/usgeeus/Pietrzak-VDF-solidity-verifier}

\section{Expected Gas Cost for Hash-to-Prime in Smart Contracts}
\label{appendix: hash-to-prime}

In Wesolowski’s VDF, implementing the hash-to-prime verification on-chain requires an algorithm capable of determining whether numerous candidate values are prime. As efficient probabilistic methods for prime checking, we considered Miller--Rabin primality test with 11 iterations, and the Baillie--PSW primality test which uses a single iteration of Miller--Rabin to base 2 and a Lucas primality test~\cite{bailliepsw}. For evaluation, we tested 6,100 randomly selected primes over \(2^{255}\). On average, verifying a single prime consumed 21,446 gas for Miller--Rabin test and 45,912 gas for Baillie--PSW test. All prime number test cases, test scripts, and results are available in the repository referenced in~\ref{appendix: code availability}.

In the worst case, the number of prime candidates can approach \(\bigl(\log x\bigr)^2\) for large \(x\), following Shanks’s upper bound on gaps between consecutive primes~\cite{Shanks1964OnMG}. Since we are dealing with the maximum number \(2^{256} - 1\), we set \(x = 2^{256}\), giving \(\bigl(\log(2^{256})\bigr)^2 \approx 5939\). Verifying 5939 candidate primes, each requiring about 21,446 gas for Miller--Rabin test or 45,912 gas for Baillie--PSW test, would quickly exceed the gas limit and render a purely on-chain implementation impractical. Moreover, even a single Miller--Rabin check can cost about \$0.28 (at 5\,gwei) to \$1.12 (at 20\,gwei) with ETH at \$2600, so multiplying by 5939 quickly runs into thousands of dollars, further underscoring the infeasibility of on-chain hash-to-prime.

\section{Proof of Theorem \ref{theorem: minimizing delta}}
\label{section: proof of minimum delta}

\begin{proof}
Firstly, we calculate the partial derivative of \(\mathcal{G}_{total}\) with respect to \(\delta\), which yields \[\frac{\partial \mathcal{G}_{total}}{\partial \delta} = \ln{2} \cdot \gamma \cdot 2 ^\delta - (\alpha+\beta).\] To show that there is a unique \(\delta\) that minimizes \(\mathcal{G}_{total}\), we need to find where this derivative equals zero. This leads us to the equation \[2^\delta = \frac{\alpha + \beta}{\ln{2} \cdot \gamma}.\]

The function \(2^\delta\) is strictly increasing. Hence, the equation \(2^\delta = \frac{\alpha + \beta}{\ln{2} \cdot \gamma} \) has a unique solution \(\delta\), which can be expressed as \[\delta = \log_2 (\alpha + \beta) - \log_2 (\ln{2} \cdot \gamma ).\] This \(\delta\) is where the first derivative equals zero.

Further, to confirm that this \(\delta\) is indeed a minimum, we need to check the second derivative. The monotonic nature of the exponential function in this context implies that as \(\delta\) increases, \(\gamma \cdot 2^\delta\) grows without bounds, while decreasing \(\delta\) below this point leads to \(\gamma \cdot 2^\delta\) becoming smaller than \(\alpha + \beta\), thereby confirming that the function is minimized at this \(\delta\).

Thus, \(\mathcal{G}_{total}(\tau, \delta)\) indeed has one unique minimum at \(\delta\) within the specified range of \(\tau > \delta \geq 0\), completing the proof.
\end{proof}

\section{Proof of Corollary \ref{corollary:integer delta}}
\label{section: proof of integer delta}

\begin{proof}
    Define
    \begin{equation}
        f(\delta) = \mathcal{G}(\tau, \delta+1) - \mathcal{G}(\tau, \delta) = \gamma \cdot (2^{\delta+1} - 2^\delta) + (\alpha + \beta) (-(\delta+1) + \delta) = \gamma \cdot 2^\delta - (\alpha+\beta) .
    \end{equation}

    Then, \(f(\delta) = 0 \) when 
    \begin{equation}
        \delta = \log_2 (\alpha+\beta) - \log_2 \gamma .
    \end{equation}
    Also, as given \(\gamma >0 , \delta > 0\)
    \begin{equation}
        \frac{df}{d\delta} = \ln 2 \cdot \gamma \cdot 2^\delta ,
    \end{equation}
    \(f\) is a monotone-increasing function in the boundary. 
    Therefore, after the point \(\delta'\) where
    \begin{equation}
        \delta'= \log_2 (\alpha+\beta) - \log_2 \gamma = \delta_m + \log_2(\ln 2),
    \end{equation}
    \(f\) is bigger than 0.

    For the 1st case of the lemma, \(\Delta < -\log_2(\ln 2)\) implies
    \begin{equation}
        f( \lfloor \delta_m \rfloor) > 0.
    \end{equation}
    As the definition of \(f\),
    \begin{equation}
        \mathcal{G}(\tau, \lfloor \delta_m \rfloor) < \mathcal{G}(\tau, \lceil \delta_m \rceil).
    \end{equation}
    In the same way, we can get the other two cases.
\end{proof}

\end{document}